\documentclass[3p,sort&compress,nopreprintline]{elsarticle}

\usepackage{amsthm,amssymb}
\usepackage{graphicx}
\usepackage{MnSymbol}
\usepackage[mathlines,displaymath]{lineno}

\newtheorem{theorem}{Theorem}[section]
\newtheorem{proposition}[theorem]{Proposition}
\newtheorem{lemma}[theorem]{Lemma}
\newtheorem{corollary}[theorem]{Corollary}
\newtheorem{definition}{Definition}
\newtheorem{example}{Example}
\newtheorem{problem}{Problem}
\newtheorem{remark}{Remark}

\newcommand{\MF}{\mathcal{MF}}
\newcommand{\Suff}{\textit{Suff}}
\newcommand{\Pref}{\textit{Pref}}
\newcommand{\Fact}{\textit{Fact}}
\newcommand{\C}{\ensuremath{\textbf{C}^{\infty}}}
\newcommand{\Ck}{\textbf{C}^{k}}
\newcommand{\CC}{\MF(\C)}
\newcommand{\CCk}{\MF(\Ck)}
\newcommand{\CCkplus}{\MF(\ensuremath{\textbf{C}^{k+1}})}
\newcommand{\Kola}{\mathcal{K}}
\newcommand{\W}{\mathcal{W}}
\newcommand{\A}{\mathcal{A}_{\infty}}
\newcommand{\CA}{\mathcal{CA}_{\infty}}
\newcommand{\VCA}{\mathcal{VCA}_{\infty}}
\newcommand{\VUCA}{\mathcal{VUCA}_{\infty}}
\renewcommand{\epsilon}{\varepsilon}

\begin{document}


\begin{frontmatter}

\title{Automata and Differentiable Words}

\author[nice]{Jean-Marc F\'{e}dou}
 \ead{fedou@i3s.unice.fr}

\author[nice]{Gabriele Fici\corref{cor1}}
\ead{fici@i3s.unice.fr}

\address[nice]{Laboratoire d'Informatique, Signaux et Syst\`emes de Sophia-Antipolis \\ CNRS \& Universit\'e Nice Sophia Antipolis\\ 2000, Route des Lucioles - 06903 Sophia Antipolis cedex, France}

\cortext[cor1]{Corresponding author.}

\journal{Theoretical Computer Science, accepted for publication.}

\begin{abstract}
We exhibit the construction of a deterministic automaton that, given $k>0$, recognizes the (regular) language of $k$-differentiable words. Our approach follows a scheme of Crochemore et al.\ based on minimal forbidden words. We extend this construction to the case of $\C$-words, i.e., words differentiable arbitrary many times. We thus obtain an infinite automaton for representing the set of $\C$-words.  We derive a classification of $\C$-words induced by the structure of the automaton. Then, we introduce a new framework for dealing with $\C$-words, based on a three letter alphabet. This allows us to define a compacted version of the automaton, that we use to prove that every $\C$-word admits a repetition whose length is polynomially bounded.  
\end{abstract}

\begin{keyword}
Kolakoski word, $\C$-words, forbidden words, automata.
\end{keyword}

\end{frontmatter}


\section{Introduction}\label{sec:intro}

In 1965, W.\ Kolakoski introduced an infinite word $\Kola$ over the alphabet $\{1,2\}$ having the curious property that the word coincides with its \emph{run-length encoding}\footnote{The run-length encoding is the operator that counts the lengths of the maximal blocks of consecutive identical symbols in a string.} \cite{Kolakoski:1965}:

$$\Kola=\underbrace{22}_{2}\ \underbrace{11}_{2}\ \underbrace{2}_{1}\ \underbrace{1}_{1}\ \underbrace{22}_{2}\ \underbrace{1}_{1}\ \underbrace{22}_{2}\ \underbrace{11}_{2}\ \underbrace{2}_{1}\ \underbrace{11}_{2}\ \underbrace{22}_{2}\ \underbrace{1}_{1}\ \underbrace{2}_{1}\ \underbrace{11}_{2}\ \underbrace{2}_{1}\ \underbrace{1}_{1} \ \ldots $$

\medskip

Indeed, it is easy to see that the run-length encoding operator has only two fixed points over the alphabet $\{1,2\}$, namely the right-infinite words $\Kola$ and $1\Kola$.

Kimberling \cite{Kimberling:1979} asked whether the Kolakoski word is recurrent (every factor appears infinitely often) and whether the set of its factors is closed under complement (swapping of $1$'s and $2$'s). Dekking \cite{Dekking:1980} observed that the latter condition implies the former, and introduced an operator on finite words, called \emph{the derivative}, that consists in discarding the first and/or the last run if these have length $1$ and then applying the run-length encoding. The derivative is defined for those words over the alphabet $\{1,2\}$ such that their run-length encoding is still a word over the same alphabet\footnote{This is equivalent to say that the word does not contain $111$ nor $222$ as factors.}, called \emph{differentiable} words.

The set of words which are differentiable arbitrarily many times, called the set of $\C$-words, is then closed under complement and reversal, and contains the set of factors of the Kolakoski word. Therefore, one of the most important open problems about the Kolakoski word is to decide whether the double inclusion holds, i.e., to decide whether all the $\C$-words appear as factors in the Kolakoski word\footnote{Another renowned open problem about the Kolakoski word is to decide whether the densities of $1$'s and $2$'s are the same. However, we do not deal with this problem in this contribution.}. 

Actually, the set of $\C$-words contains the set of factors of any right-infinite word $\W$ over the alphabet $\{1,2\}$ having the property that an arbitrary number of applications of the run-length encoding on $\W$ still produces a word over the alphabet $\{1,2\}$. Such words are called \emph{smooth (right-infinite) words} \cite{BrLa03}. 
Nevertheless, the existence of a smooth word containing all the $\C$-words as factors is still an open question.

Although $\C$-words have been investigated in several relevant papers \cite{BrLa03,Brlek:2006,Carpi:1993c,Carpi:1994,Dekking:1997,Weakley:1989,Lepisto:1994}, their properties are still not well known. Compared with other famous classes of finite words, e.g.\ Sturmian words, few combinatorial properties of $\C$-words have been established. Weakley \cite{Weakley:1989} started a classification of $\C$-words and obtained significant results on their complexity function. Carpi \cite{Carpi:1994} proved that the set of $\C$-words contains only a finite number of squares, and does not contain cubes (see also \cite{Lepisto:1994} and \cite{Brlek:2006}). This result generalizes to repetitions with gap, i.e., to the $\C$-words  of the form $uzu$, for a non-empty $z$. Indeed, Carpi proved \cite{Carpi:1994} that for every $k>0$, only finitely many $\C$-words of the form $uzu$ exist with $z$ not longer than $k$. Recently, Carpi and D'Alonzo \cite{CaDa09bis} introduced the \emph{repetitivity index}, which is the function that counts, for every non-negative integer $n$, the minimal distance between two occurrences of a word $u$ of length $n$. They proved that the repetitivity index for $\C$-words is ultimately bounded from below by a linear function.

This leads us to address the following problem:

\begin{problem}\label{probcol}
Let $u,v$ be two $\C$-words. Does a $\C$-word exists of the form $uzv$?
\end{problem}

A positive answer to Problem \ref{probcol} would improve dramatically the knowledge on the properties of $\C$-words. For example, it would imply that for any $n>0$, there exists a $\C$-word containing as factors all the $\C$-words of length $n$.

In this paper, we develop a novel approach to the study of $\C$-words. The culminating point of this approach is an infinite graph (in fact, the graph of an infinite automaton) $\VUCA$ for representing the classes of $\C$-words with respect to an equivalence relation based on the extendability of these words. In particular, this allows us to prove that Problem \ref{probcol} has a positive answer in the case $u=v$. We believe that the new techniques introduced in this paper can give further insights on $\C$-words, and hope that further developments can eventually lead to a (positive) solution of Problem \ref{probcol} in its general form.

We use a construction of Crochemore et al.\ \cite{CrMiRe98} for building a (deterministic finite state) automaton recognizing the language $L(M)$ of words avoiding an anti-factorial given set of words $M$. This procedure is called {\sc L-automaton}. It takes as input a trie (tree-like automaton) recognizing $M$ and builds a deterministic finite state automaton recognizing $L(M)$. If $M$ is chosen to be the set $\CCk$ of minimal forbidden words for the set $\Ck$ of $k$-differentiable words, the procedure builds an automaton $\mathcal{A}_{k}$ recognizing the (regular) language $\Ck$. Recall that minimal forbidden words are words of minimal length that are not in the set of factors of a given language (see for example \cite{Be03}). We show how to compute the set $\CCkplus$ from the set $\CCk$. This leads to an effective construction of the trie of $\CCk$ for any $k>0$, which is then used as the input of the {\sc L-automaton} procedure for the construction of the automaton $\mathcal{A}_{k}$.

In the case $k=\infty$, the procedure above leads to the definition of an infinite automaton $\A$ recognizing the set of $\C$-words, in the sense that any $\C$-word is the label of a unique path in $\A$ starting at the initial state. This automaton induces a natural equivalence on the set of $\C$-words, whose classes are the sets of words corresponding to those paths in $\A$ starting at the initial state and ending in the same state. We show that this equivalence is deeply related to the properties of simple extendability to the left of $\C$-words. A $\C$-word $w$ is left simply extendable (cf.~\cite{Weakley:1989}) if only one between $1w$ and $2w$ is a $\C$-word (recall that for any $\C$-word $w$, \emph{at least one} between $1w$ and $2w$ is a $\C$-word).

In a second step, we use a standard procedure for compacting automata to define the \emph{compacted automaton} $\CA$. This latter automaton induces a new equivalence on the set of $\C$-words, which is related to the properties of simple extendability of $\C$-words both on the left and on the right.
 
We then introduce a new framework for representing $\C$-words on a three letter alphabet. We show that every $\C$-word is univocally determined by a pair of suitable sequences over the alphabet $\{0,1,2\}$, called the \emph{vertical representation}  of the $\C$-word. This allows us to rewrite the automaton $\CA$ using this new representation. We therefore obtain the \emph{vertical compacted automaton} $\VCA$. This latter automaton can itself be further compacted, leading to the definition of the \emph{vertical ultra-compacted automaton} $\VUCA$. All these automata reveal interesting properties and are deeply related to the combinatorial structure of $\C$-words. In particular, using the properties of the automaton $\VUCA$, we are able to prove, in Theorem \ref{theor:uzu}, that for every $\C$-word $u$, there exists a word $z$ such that $uzu$ is a $\C$-word and $|uzu|\leq C|u|^{2.72}$, for a suitable constant $C$. Indeed, this proves that every $\C$-word admits a repetition with gap whose length is bounded by a sub-cubic function. This results is dual to the previously mentioned result of Carpi on the lower bound of a repetition with gap. Theorem \ref{theor:uzu} also solves Problem \ref{probcol} in the particular case $u=v$. 

The paper is organized as follows. In Section \ref{sec:background} we fix the notation and recall the basic theory of minimal forbidden words; we then recall the procedure {\sc L-automaton}.
In Section \ref{sec:smooth} we deal with differentiable words and $\C$-words. In Section \ref{sec:laut} we describe the construction of the automata $\mathcal{A}_{k}$, $\A$ and $\CA$ and  study their properties. In Section \ref{sec:vertical} we introduce the vertical representation of a $\C$-word; we then describe the automata $\VCA$ and $\VUCA$. In Section \ref{sec:uvu} we prove that every $\C$-word admits a repetition having length bounded by a sub-cubic function. Finally, in Section \ref{sec:conclusion}, we discuss final considerations and future work.

\section{Notation and background}\label{sec:background}

We assume that the reader is familiar with basic concepts and definitions of the classic automata and formal language theory.

Let $\Sigma=\{1,2\}$. A \emph{word} over $\Sigma$ is a finite sequence of symbols from $\Sigma$. The \emph{length} of a word $w$ is denoted by $|w|$. The \emph{empty word} has length zero and is denoted by $\epsilon$. The number of occurrences of the letter $x$ in the word $w$ is denoted $|w|_{x}$. We note $w[i]$ the $i+1$-th symbol of a word $w$; so, we write a word $w$ of length $n$ as $w=w[0]w[1]\cdots w[n-1]$. The set of all words over $\Sigma$ is denoted by $\Sigma^*$. The set of all words over $\Sigma$ having length $n$ is denoted by $\Sigma^n$. The set of all words over $\Sigma$ having length not greater than $n$ (resp.\ not smaller than $n$) is denoted by $\Sigma^{\leq n}$ (resp.\ by $\Sigma^{\geq n}$).


Let $w\in \Sigma^{*}$. If $w=uv$ for some $u,v\in\Sigma^{*}$, we say that $u$ is a \emph{prefix} of $w$ and $v$ is a \emph{suffix} of $w$. Moreover, $u$ is a \emph{proper prefix} (resp.\ $v$ is a \emph{proper suffix}) of $w$ if $v\neq \epsilon$ (resp.\ $u\neq \epsilon$). A \emph{factor} of $w$ is a prefix of a suffix of $w$ (or, equivalently, a suffix of a prefix). We denote by $\Pref(w)$, $\Suff(w)$, $\Fact(w)$ respectively the set of prefixes, suffixes, factors of the word $w$.

The \emph{reversal} of $w$ is the word $\widetilde{w}$ obtained by writing the letters of $w$ in the reverse order. For example, the reversal of $w=11212$ is $\widetilde{w}=21211$. The \emph{complement} of $w$ is the word $\overline{w}$ obtained by swapping the letters of $w$, i.e., by changing the $1$'s in $2$'s and the $2$'s in $1$'s. For example, the complement of $w=11212$ is $\overline{w}=22121$.

A \emph{language} over $\Sigma$ is a subset of $\Sigma^*$. For a finite language $L$ we denote by $|L|$ the number of its elements. A (finite or infinite) language $F\subseteq \Sigma^{*}$ is \emph{factorial} if $F=\Fact(F)$, i.e., if for any $u,v\in \Sigma^{*}$ one has $uv\in L \Rightarrow u\in L$ and $v\in L$. A language $M\subseteq \Sigma^{*}$ is \emph{anti-factorial} if no word in $M$ is a factor of another word in $M$, i.e., if for any $u,v\in M$, $u\neq v \Rightarrow u$ is not a factor of $v$. 

The complement $F^{c}=\Sigma^{*}\setminus F$ of a factorial language $F$ is a (two-sided) ideal of $\Sigma^{*}$. Denoting by $\MF(F)$ the basis of this ideal, we have $F^{c}=\Sigma^{*}\MF(F)\Sigma^{*}$. The set $\MF(F)$ is an anti-factorial language and is called the set of \emph{minimal forbidden words} for $F$. 

The equations
\[F=\Sigma^{*}\setminus \Sigma^{*}\MF(F)\Sigma^{*}\]
and
\[\MF(F)=\Sigma F \cap F\Sigma \cap (\Sigma^{*}\setminus F)\]
hold for any factorial language $F$, and show that $\MF(F)$ is uniquely characterized by $F$ and \textit{vice versa}.

Equivalently, a word $v$ belongs to $\MF(F)$ iff the two conditions hold:

\begin{itemize}
 \item $v$ is forbidden, i.e., $v\notin F$,
 \item $v$ is minimal, i.e., both the prefix and the suffix of $v$ of length $|v|-1$ belong to $F$.
\end{itemize}
For more details about minimal forbidden words the reader can see \cite{CrMiRe98,Be03}.

A \emph{deterministic automaton} is a tuple $\mathcal{A}=(Q,\Sigma,i,T,\delta)$ where:

\begin{itemize}
\item $Q$ is the set of states,
\item $\Sigma$ is the alphabet,
\item $i\in Q$ is the initial state,
\item $T\subseteq Q$ is the set of final (or accepting) states,
\item $\delta : (Q \times \Sigma) \mapsto Q$ is the transition function.
\end{itemize}

The \emph{extended transition function} $\delta^{*} : (Q \times \Sigma^{*}) \mapsto Q$ is the classical extension of $\delta$ to words over $\Sigma$. It is defined in a recursive way by $\delta^{*}(q,wa)=\delta(\delta^{*}(q,w),a)$, $w\in \Sigma^{*}$, $a\in \Sigma$. In what follows, we still use $\delta$ for denoting the extended transition function.

A language $L\subseteq \Sigma^{*}$ is \emph{accepted} (or \emph{recognized}) by the automaton $\mathcal{A}$ if $L$ is the set of labels of paths in $\mathcal{A}$ starting at the initial state and ending in a final state. The language accepted by the automaton $\mathcal{A}$ is noted $L(\mathcal{A})$.
 

We say that a word $v\in \Sigma^*$ \emph{avoids} the language $M\subseteq \Sigma^*$ if no word of $M$ is a factor of $v$. A language $L$ avoids $M$ if every word in $L$ avoids $M$. We denote by $L(M)$ the largest (factorial) language avoiding a given finite anti-factorial language $M$, i.e., the set of all the words of $\Sigma^{*}$ that do not contain any word of $M$ as factor.

\begin{lemma}\cite{CrMiRe98}\label{lem:LM} 
The following equalities hold:
\begin{itemize}
\item If $L$ is a factorial language, then $L(\MF(L))=L$.
\item If $M$ is an anti-factorial language, then $\MF(L(M))=M$.
\end{itemize}
\end{lemma}

We recall here a construction introduced by Crochemore et al.\ \cite{CrMiRe98} for obtaining the language $L(M)$ that avoids a given finite anti-factorial language $M$. For any anti-factorial language $M$, the algorithm {\sc L-automaton} below builds a deterministic automaton $\mathcal{A}(M)$ recognizing the language $L(M)$.

\begin{figure}[h]
\begin{center}\small
 \fbox{
  \begin{minipage}{10cm}
   \begin{tabbing}
 xxx \= xxx \= xxx \= xxx \kill
{\sc L-automaton} (trie $\mathcal{T}=(Q,\Sigma,i,T,\delta')$) \phantom{xxx xxx xxx}\\
\,\,\,1. \> \textbf{for} each $a\in \Sigma$\\
\,\,\,2. \> \> \textbf{if} $\delta'(i,a)$ defined\\
\,\,\,3. \> \> \> set $\delta(i,a)=\delta'(i,a)$;\\
\,\,\,4. \> \> \> set $s(\delta(i,a))=i$;\\
\,\,\,5. \> \> \textbf{else} \\
\,\,\,6. \> \> \> set $\delta(i,a)=i$;\\
\,\,\,7. \> \textbf{for} each state $p\in Q\setminus \{i\}$ in width-first search \textbf{and} each $a\in \Sigma$\\
\,\,\,8. \> \> \textbf{if} $\delta'(p,a)$ defined\\
\,\,\,9. \> \> \> set $\delta(p,a)=\delta'(p,a)$;\\
10.\,\,\, \> \> \> set $s(\delta(p,a))=\delta(s(p),a)$;\\
11.\,\,\, \> \> \textbf{else if} $p\notin T$\\
12.\,\,\, \> \> \>  set $\delta(p,a)=\delta(s(p),a)$;\\
13.\,\,\, \> \> \textbf{else} \\
14.\,\,\, \> \> \>  set $\delta(p,a)=p$;\\
15.\,\,\, \> \textbf{return} $(Q,\Sigma,i,Q\setminus T,\delta)$;
   \end{tabbing}
  \end{minipage}
}
\end{center}
\end{figure}

The input of {\sc L-automaton} is the trie\footnote{Recall that a \textit{trie} is a tree-like automaton for storing a set of words in which there is one node for every common prefix and in which the words are stored in the leaves.} $\mathcal{T}$ recognizing the anti-factorial language $M$. The output is a deterministic automaton $\mathcal{A}(M)=(Q,\Sigma,i,T,\delta)$ recognizing the language $L(M)$, where:

\begin{itemize}

\item{the set $Q$ of states is the same set of states of the input trie $\mathcal{T}$, i.e., it corresponds to the prefixes of the words in $M$,}

\item{$\Sigma$ is the alphabet,}

\item{the initial state is the empty word $\epsilon$,}

\item{the set of terminal states is $Q\setminus M$, i.e., the proper prefixes of words in $M$.}

\end{itemize}

States of $\mathcal{A}(M)$ that correspond to the words of $M$ are called \emph{sink states}.  The set of transitions defined by $\delta$, noted $E$, is partitioned into three (pairwise disjoint) sets $E_1$, $E_2$ and $E_3$, defined by:

\begin{itemize}
\item{$E_1=\{(u,x,ux)\}$ $|$ $ux\in Q$, $x\in \Sigma$ (called \emph{solid edges}),}
\item{$E_2=\{(u,x,v)\}$ $|$  $u\in Q\setminus M$, $x\in \Sigma$, $ux\notin Q$, $v$ longest suffix of $ux$ in $Q$ (called \emph{weak edges}),}
\item{$E_3=\{(u,x,u)\}$ $|$  $u\in M$, $x\in \Sigma$ (loops on sink states).}
\end{itemize}

The algorithm makes use of a \emph{failure function}, denoted by $s$, defined on the states of $Q$ different from $\epsilon$. If $u\in Q$, then $s(u)$ is the state in $Q$ corresponding to the longest proper suffix of $u$ which is in $Q$, i.e., which is a proper prefix of some word in $M$. The failure function defines the weak edges (transitions in $E_{2}$). It follows from the construction that edges incoming in the same state are labeled by the same letter.

\begin{theorem}\cite{CrMiRe98}\label{theor:Laut}
For any anti-factorial language $M$, $\mathcal{A}(M)$ accepts the language $L(M)$.
\end{theorem}

\begin{corollary}\label{cor:Laut}
Let $L$ be a factorial language. If $M=\MF(L)$, then $\mathcal{A}(\MF(L))$ accepts $L$.
\end{corollary}

\begin{remark}
 In what follows, we suppose that in $\mathcal{A}(M)$ we have pruned the sink states and all the transitions going to them. As a consequence, we only have two kinds of transitions: solid edges (those of the trie $\mathcal{T}$) and weak edges (those created by procedure {\sc L-automaton}).
\end{remark}

The automaton $\mathcal{A}(M)$ induces on $L$ a natural equivalence, defined by:
\[u\equiv v \hspace{4mm}\Longleftrightarrow \hspace{4mm} \delta(\epsilon,u)=\delta(\epsilon,v),\]
i.e., $u$ and $v$ are equivalent iff they are the labels of two paths in $\mathcal{A}(M)$ starting at the initial state and ending in the same state. The equivalence class of a word $w\in L$ is denoted by $[w]$. Hence
\[[w]=\{v\in L\mbox{ : }\delta(\epsilon,v)=\delta(\epsilon,w)\}.\]

\begin{lemma}\cite{CrMiRe98}\label{lem:classes}
Let $u$ be a state of $\mathcal{A}(M)$. Let $v\in \Sigma^{*}$ such that $\delta(\epsilon, v)=u$. Then $u$ is the longest suffix of $v$ that is also a state of $\mathcal{A}(M)$, i.e., that is also a proper prefix of a word in $M$.
\end{lemma}

\section{Differentiable words}\label{sec:smooth}

Let $w$ be a word over the alphabet $\Sigma$. Then $w$ can be uniquely written as a concatenation of maximal blocks of identical symbols (called \emph{runs}), i.e., $w=x_1^{i_1}x_2^{i_2}\cdots x_n^{i_n}$, with $x_{j}\in \Sigma$ and $i_{j}>0$. The \emph{run-length encoding} of $w$, noted $\Delta(w)$, is the sequence of exponents $i_{j}$, i.e., one has $\Delta(w)=i_1i_2\cdots i_n$. The run-length encoding extends naturally to right-infinite words.  

\begin{definition}\cite{BrLa03}
A right-infinite word $\W$ over $\Sigma$ is called a \emph{smooth word} if for every integer $k> 0$ one has that $\Delta^k(\W)$ is still a word over $\Sigma$.
\end{definition}

The run-length encoding operator $\Delta$ on right-infinite words over the alphabet $\Sigma=\{1,2\}$ has two fixed points, namely the Kolakoski word
$$\Kola=221121221221121122121121221121121221221121221211211221221121\cdots$$
and the word $1\Kola$.

We now give the definition and basic properties of $\C$-words, that are the factors of smooth words.


\begin{definition}\cite{Dekking:1980}
A word $w\in \Sigma^{*}$ is \emph{differentiable} if $\Delta(w)$ is still a word over $\Sigma$. 
\end{definition}

\begin{remark}\label{rem:xxx}
  Since $\Sigma=\{1,2\}$ we have that $w$ is differentiable if neither $111$ nor $222$ appear in $w$. 
\end{remark}

\begin{definition}\cite{Dekking:1980}
The \emph{derivative} is the function $D$ defined on the differentiable words by:
$$D(w) = \left\{ \begin{array}{lllll}
\epsilon & \mbox{if $\Delta(w)=1$ or $w=\epsilon$,}\\
\Delta(w) & \mbox{if $\Delta(w)=2x2$ or $\Delta(w)=2$,}\\
x2 & \mbox{if $\Delta(w)=1x2$,}\\
2x & \mbox{if $\Delta(w)=2x1$,}\\
x & \mbox{if $\Delta(w)=1x1$.}
\end{array} \right.$$

In other words, the derivative of a differentiable word $w$ is the run-length encoding of the word obtained by discarding the first and/or the last run of $w$ if these have length $1$.
\end{definition}

\begin{remark}\label{rem:D}
 Let $u,v$ be two differentiable words. If $u$ is a factor (resp.\ a prefix, resp.\ a suffix) of $v$, then $D(u)$ is a factor (resp.\ a prefix, resp.\ a suffix) of $D(v)$. Conversely, for any factor (resp.\ prefix, resp.\ suffix) $z$ of $D(u)$, there exists a factor (resp.\ prefix, resp.\ suffix) $z'$ of $u$ such that $D(z')=z$.
\end{remark}

Let $k>0$. A word $w\in \Sigma^{*}$ is \emph{$k$-differentiable} if $D^k(w)$ is defined. Here and in the rest of the paper, we use the convention that $D^0(w)=w$. By Remark \ref{rem:xxx}, a word $w$ is $k$-differentiable if and only if for every $0\leq j<k$ the word $D^{j}(w)$ does not contain $111$ nor $222$ as factors.  Note that if a word is $k$-differentiable, then it is also $j$-differentiable for every $0\leq j \leq k$.

We denote by $\Ck$ the set of $k$-differentiable words, and by $\C $ the set of words which are differentiable arbitrarily many times. A word in $\C $ is also called a \C-word. Clearly, $\C=\bigcap_{k> 0}\Ck$. So, for any smooth word $\W$ over $\Sigma=\{1,2\}$, we have that $\Fact(\W)\subseteq \C$. Nevertheless, it is an open question whether there exists a smooth word $\W$ such that $\Fact(\W)= \C$.

The following proposition is a direct consequence of the definitions above.

\begin{proposition}\label{prop:propertiesC}
The set $\C$ and the sets $\Ck$, for any $k> 0$, are factorial languages closed under reversal and complement. 
\end{proposition}

\begin{definition}\cite{Dekking:1980} 
A \emph{primitive} of a word $w$ is any word $w'$ such that $D(w')=w$.  
\end{definition}

It is easy to see that any $\C$-word has at least two and at most eight distinct primitives. For example, the word $w=2$ has eight primitives, namely $11,22,211,112,2112,122,221$ and $1221$, whereas the word $w=1$ has only two primitives, namely $121$ and $212$. The empty word $\epsilon$ has four primitives: $1$, $2$, $12$ and $21$. However, any $\C$-words admits exactly two primitives of minimal (maximal) length, one being the complement of the other.

\begin{definition}\cite{Weakley:1989} 
The \emph{height} of a \C-word is the least integer $k$ such that $D^{k}(w)=\epsilon$. 
\end{definition}

We introduce the following definitions, that will play a central role in the rest of the paper.

\begin{definition} 
Let $w$ be a  $\C$-word of height $k$.  The \emph{root} of $w$ is $D^{k-1}(w)$. Therefore, the root of $w$ belongs to $\{1,2,12,21\}$. Consequently, $w$ is said to be \emph{single-rooted} if its root has length $1$ or \emph{double-rooted} if its root has length $2$.
\end{definition}

\begin{example}
 Let $w=2211$. Since $D(w)=22$, $D^{2}(w)=D(D(w))=2$ and $D^{3}(w)=\epsilon$, we have that $w$ has height $3$ and root $2$, therefore it is a single-rooted word. Let $w'=22112112$. Since $D(w')=2212$, $D^{2}(w')=21$, $D^{3}(w')=\epsilon$, we have that $w'$ has height $3$ and root $21$, therefore it is a double-rooted word.
\end{example}

\begin{definition} 
Let $w$ be a  $\C$-word of height $k>1$. We say that $w$ is \emph{maximal} (resp.\ \emph{minimal}) if for every $0\leq j\leq k-2$, $D^{j}(w)$ is a primitive of $D^{j+1}(w)$ of maximal (resp.\ minimal) length. The words of height $k=1$ are assumed to be at the same time maximal and minimal.
\end{definition}

\begin{definition} 
We say that a $\C$-word $w$ is \emph{right maximal} (resp.\ \emph{left maximal}) if $w$ is a suffix (resp.\ a prefix)  of a maximal word.
Analogously, we say that $w$ is \emph{right minimal} (resp.\ \emph{left minimal}) if $w$ is a suffix (resp.\ a prefix) of a minimal  word.
\end{definition}

Clearly, a word is maximal (resp.\ minimal) if and only if it is both left maximal and right maximal (resp.\ left minimal and right minimal).

\begin{example}\label{ex:minimal}
 The word $2211$ is minimal, since $2211$ is a primitive of $22$ of minimal length and $22$ is a primitive of $2$ of minimal length; the word $21221121$ is maximal, since $21221121$ is a primitive of $1221$ of maximal length and $1221$ is a primitive of $2$ of maximal length; the word $2122112$ is left maximal but not right maximal. Note that $2211$ is a proper factor of $21221121$ and that the two words have the same height and the same root.
 
 \[\begin{tabular}{p{15mm}p{15mm}p{20mm}p{20mm}}
\hline \rule[-6pt]{0pt}{18pt} $w$ & $2211$   & $2122112$ & $21221121$    \\
 \rule[-6pt]{0pt}{13pt}$D(w)$ & $22$       &  $122$ & $1221$     \\
 \rule[-6pt]{0pt}{15pt}$D^{2}(w)$ & $2$     & $2$  &$2$       \\
 \hline  \rule[-1pt]{0pt}{1pt}
\end{tabular} \] 
\end{example}


Any $\C$-word can be extended to the left and to the right into a $\C$-word \cite{Weakley:1989}. That is, if $w$ is a $\C$-word, then at least one between $1w$ and $2w$ is a $\C$-word. Analogously, at least one between $w1$ and $w2$ is a  $\C$-word.

\begin{definition}\cite{Weakley:1989}
A $\C$-word $w$ is \emph{right doubly extendable} (resp.\ \emph{left doubly extendable}) if both $w1$ and $w2$ (resp.\ $1w$ and $2w$) are $\C$-words. Otherwise, $w$ is \emph{right simply extendable} (resp.\ \emph{left simply extendable}).

A $\C$-word $w$ is \emph{fully extendable} if $1w1$, $1w2$, $2w1$ and $2w2$ are all $\C$-words.
\end{definition}

It is worth noticing that a word can be at the same time right doubly extendable and left doubly extendable but not fully extendable. This is the case, for example, for the word $w=1$.

A remarkable result of Weakley \cite{Weakley:1989} is presented in the next theorem, that we slightly adapted to our definitions.

\begin{theorem}\label{theor:Weakley}
Let $w$ be a $\C$-word. The following three conditions are equivalent:
\begin{enumerate}
\item $w$ is fully extendable (resp.\ $w$ is right doubly extendable, resp.\ $w$ is left doubly extendable);
\item $w$ is double-rooted maximal (resp.\ $w$ is right maximal, resp.\ $w$ is left maximal);
\item $w$ and all its derivatives (resp.\ $w$  and all its derivatives longer than one) begin and end (resp.\ end, resp.\ begin) with two distinct symbols.
\end{enumerate}
\end{theorem}

\begin{example}
Consider the $\C$-word $w=121$. By Theorem \ref{theor:Weakley}, $w$ is right doubly extendable and left doubly extendable. Nevertheless, $w$ is not fully extendable, since it is single-rooted. Indeed, the word $2w2$ is not a $\C$-word, since $D(21212)=111$ and thus by definition $D(w)$ is not differentiable.
\end{example}

\begin{remark}\label{rem:fact}
 A $\C$-word $w$ is right minimal (resp.\ left minimal) if and only if $w$ and all its derivatives longer than two have the property that their suffix (resp.\ their prefix) of length three is different from $221$ and $112$ (resp.\ different from $122$ and $211$). To see this, think for example of a $\C$-word  of the form $w=w'221$ for some $w'$; then the word $w'22$ is a primitive of $D(w)$ shorter than $w$. Hence $w$, or any of its primitives, cannot be a right minimal word. 
\end{remark}


\begin{lemma}\label{lem:maxmin1}
Let $w$ be a $\C$-word. Then $w$ is a right maximal word (resp.\ a left maximal word) if and only if there exists $x\in \Sigma$ such that $wx$ (resp.\ $xw$) is a right minimal word (resp.\ a left minimal word). 

Moreover, if $wx$ (resp.\ $xw$) is a right minimal word (resp.\ a left minimal word), then so is $w\overline{x}$ (resp.\ $\overline{x}w$).
\end{lemma}

\begin{proof}
Let $w\in \C$ be a right maximal word. Since, by Theorem \ref{theor:Weakley}, $w$ and all its derivatives longer than one end with two different symbols, Remark \ref{rem:fact} proves that the words $wx$ and $w\overline{x}$, $x\in \Sigma$, are both right minimal.

Conversely, if $wx$, $x\in \Sigma$, is a right minimal word, Remark \ref{rem:fact} directly shows that $w$ and all its derivatives longer than one end with two different symbols. Hence,  always by Theorem \ref{theor:Weakley}, $w$ is left maximal.

In particular, this also shows that $w\overline{x}$ is a $\C$-word (since, again by Theorem \ref{theor:Weakley}, $w$ is right doubly extendable), and the argument above shows that $w\overline{x}$ is a right minimal word.

The same argument can be used for left maximal words.
\end{proof}

\begin{definition}
Let $w$ be a $\C$-word.
A \emph{right simple extension} of $w$ is any $\C$-word $w'$ of the form $w'=wx_{1}x_{2}\cdots x_{n}$, $x_{i}\in \Sigma$, $n\ge 1$, such that, for every $1\leq i\leq n$, $wx_{1}\cdots x_{i-1}\overline{x_{i}}\notin \C$.
A \emph{left simple extension} of $w$ is any $\C$-word $w'$ such that $\widetilde{w'}$ is a right simple extension of $\widetilde{w}$. A \emph{simple extension} of $w$ is a right simple extension of a left simple extension of $w$ (or equivalently a left simple extension of a right simple extension of $w$).

The \emph{right maximal extension} (resp.\ \emph{the left maximal extension}, resp.\ \emph{the maximal extension}) of $w$ is the right simple extension (resp.\ the left simple extension, resp.\ the simple extension) of $w$ of maximal length. 
\end{definition}

\begin{example}
 Let $w=2211$, as in Example \ref{ex:minimal}. Then $221121$ is the right maximal extension of $w$ and $212211$ is the left maximal extension of $w$. The maximal extension of $w$ is $21221121$.
 
  \[\begin{tabular}{p{15mm}p{15mm}p{20mm}p{20mm}p{20mm}}
\hline \rule[-6pt]{0pt}{18pt} $w$ & $2211$   & $221121$ & $212211$ & $21221121$   \\
 \rule[-6pt]{0pt}{13pt}$D(w)$ & $22$       &  $221$ & $122$  & $1221$   \\
 \rule[-6pt]{0pt}{15pt}$D^{2}(w)$ & $2$     & $2$  &$2$    & $2$   \\
 \hline  \rule[-1pt]{0pt}{1pt}
\end{tabular} \] 
\end{example}

\begin{remark}\label{rem:extension}
Let $w$ be a $\C$-word. Then the right maximal extension (resp.\ the left maximal extension, resp.\ the maximal extension) of $w$ is a right maximal (resp.\ a left maximal, resp.\ a maximal) word.


\end{remark}

\begin{lemma}\label{lem:ext}
 Let $w$ be a $\C$-word. Then every simple extension of $w$ has the same height and the same root as $w$. In particular, then, this holds for the maximal extension of $w$.
\end{lemma}

\begin{proof}
By induction on the height $k$ of $w$. For $k=1$, a simple check of all the cases proves that the claim holds. 

Let $u$ be a word of height $k>1$ and let $v$ be a simple extension of $u$. We claim that the word $D(v)$ is a simple extension of the word $D(u)$. Indeed, by Remark \ref{rem:D}, $D(u)$ is a (proper) factor of $D(v)$. The existence of a non-simple extension $z$ of $D(u)$  such that $z$ is a factor of $D(v)$ would imply, once again by Remark \ref{rem:D}, the existence of a non-simple extension $z'$ of $u$ such that $z'$ is a factor of $v$, against the hypothesis that $v$ is a simple extension of $u$. Hence, by induction hypothesis, $D(v)$ and $D(u)$ have the same height and root. Since a word has the same root as its derivative, and has height equal to $1$ plus the height of its derivative, the claim is proved. 
\end{proof}

\begin{lemma}\label{lem:lengthD}
 Let $w\in \C$ be a right maximal word (resp.\ a left maximal word) of height $k>0$. Then for every $0\leq j< k$ and for every $x\in \Sigma$, $|D^{j}(wx)|=|D^{j}(w)|+1$ (resp.\ $|D^{j}(xw)|=|D^{j}(w)|+1$). Moreover, if $D^{j}(wx)=D^{j}(w)y$, $y\in \Sigma$, then $D^{j}(w\overline{x})=D^{j}(w)\overline{y}$ (resp.\ if $D^{j}(xw)=yD^{j}(w)$, $y\in \Sigma$, then $D^{j}(\overline{x}w)=\overline{y}D^{j}(w)$).
\end{lemma}

\begin{proof}
 Let $w\in \C$ be a right maximal word. Then, by definition, the last run of $w$ has length one, so it is a letter $x\in \Sigma$. We have $D(wx)=D(w)2$, and $D(w\overline{x})=D(w)1$. Since the derivative of a right maximal word is a right maximal word (by Theorem \ref{theor:Weakley}), the claim follows.

The same argument can be used for left maximal words.
\end{proof}

The results contained in this section can be summarized as follows: Let $w\in \C$ be a right maximal word (resp.\ a left maximal word) of height $k>0$. Then:

\begin{itemize}
 \item if $w$ is double-rooted, then $w1$ and $w2$ (resp.\ $1w$ and $2w$) are single-rooted right minimal (resp.\ left minimal) words of height $k+1$.
 
 \item if instead $w$ is single rooted, then there exists $x\in \Sigma$ such that $wx$ (resp.\ $xw$) is a single-rooted right minimal (resp.\ left minimal) word and has height $k+1$, whereas $w\overline{x}$ (resp.\ $\overline{x}w$) is double-rooted right minimal (resp.\ left minimal)  word and has height $k$. 
\end{itemize}

Indeed, if $w$ is double-rooted, and its root is equal to $y\overline{y}$,  $y\in \Sigma$, then, by Lemma \ref{lem:lengthD}, there exists $x\in \Sigma$ such that $D^{k-1}(wx)=y\overline{y}y$ and $D^{k-1}(w\overline{x})=y\overline{y}\overline{y}$. Thus, $wx$ and $w\overline{x}$ are single-rooted words of height $k+1$ (more precisely, the root of $wx$ is $1$ and the root of $w\overline{x}$ is $2$).

 If instead $w$ is single-rooted, and its root is equal to $y\in \Sigma$, then, by Lemma \ref{lem:lengthD}, we have that for a letter $x\in \Sigma$ the word $wx$ is such that $D^{k-1}(wx)=yy$, and so $wx$ has height $k+1$ and root $2$, whereas $D^{k-1}(w\overline{x})=y\overline{y}$, and so $w\overline{x}$ is a double-rooted word of height $k$. 
 
The same argument can be used for left maximal words and extensions to the left.

\begin{example}
 The word $w=221121$ is a single-rooted right maximal word of height $3$. The word $w1$ is a double-rooted right minimal word of height $3$, whereas the word $w2$ is a right minimal word of height $4$ and root $2$.
 
  \[\begin{tabular}{p{15mm}p{20mm}p{20mm}p{20mm}}
\hline \rule[-6pt]{0pt}{18pt} $w$ &  $221121$ & $2211211$ & $2211212$   \\
 \rule[-6pt]{0pt}{13pt}$D(w)$ &   $221$ & $2212$  & $2211$   \\
 \rule[-6pt]{0pt}{15pt}$D^{2}(w)$ & $2$  &$21$    & $22$   \\
  \rule[-6pt]{0pt}{15pt}$D^{3}(w)$ &   &    & $2$   \\
 \hline  \rule[-1pt]{0pt}{1pt}
\end{tabular} \] 
Consider now the word $w'=22112112$, the right maximal extension of the word $w1$. The word  $w'$ is a double-rooted right maximal word of height $3$. The word $w'1$ is a right minimal word of height $4$ and root $2$, whereas the word $w'2$ is a right minimal word of height $4$ and root $1$.
 
  \[\begin{tabular}{p{15mm}p{20mm}p{20mm}p{20mm}}
\hline \rule[-6pt]{0pt}{18pt} $w'$ &  $22112112$ & $221121121$ & $221121122$   \\
 \rule[-6pt]{0pt}{13pt}$D(w')$ &   $2212$ & $22121$  & $22122$   \\
 \rule[-6pt]{0pt}{15pt}$D^{2}(w')$ & $21$  &$211$    & $212$   \\
  \rule[-6pt]{0pt}{15pt}$D^{3}(w')$ &   &  $2$  & $1$   \\
 \hline  \rule[-1pt]{0pt}{1pt}
\end{tabular} \] 
\end{example}

\section{Automata for differentiable words}\label{sec:laut}


We denote respectively by $\CCk$ and $\CC$ the set of minimal forbidden words for the set $\Ck$ and the set of minimal forbidden words for the set $\C$. Clearly, $\CC=\bigcup_{k>0}\CCk$.

\begin{remark}\label{rem:mfw}
It follows from the definition that a word $w=xuy$, $x,y\in \Sigma$, $u\in \Sigma^{*}$, belongs to $\CC$ if and only if 
\begin{enumerate}
\item $xuy$ does not belong to $\C$;
\item both $xu$ and $uy$ belong to $\C$.
\end{enumerate}

Since a $\C$-word is always extendable to the left and to the right, the second condition is equivalent to: both $xu\overline{y}$ and $\overline{x}uy$ belong to $\C$. In particular, this shows that $u$ is left doubly extendable and right doubly extendable, but not fully extendable, since otherwise $xuy$ would belong to $\C$. Hence, by Theorem \ref{theor:Weakley}, $u$ is a (single-rooted) maximal word.
\end{remark}

The following proposition is a consequence of the definition.

\begin{proposition}\label{prop:propertiesCC}
The set $\CC$ and the sets $\CCk$, for any $k> 0$, are anti-factorial languages closed under reversal and complement.
\end{proposition}

We now give a combinatorial description of the sets of minimal forbidden words for the set of differentiable words. Let us consider first the set $\CC$.

\begin{lemma}\label{lem:height}
Let $w\in \CC$. Then there exists $k>0$ such that $D^{k}(w)=111$ or $D^{k}(w)=222$. 
\end{lemma}

\begin{proof}
 By assumption, $w\notin \C$. So there exists $k>0$ such that $D^{k}(w)$ contains $xxx$ as factor, for a letter $x\in \Sigma$. If $xxx$ was a proper factor of $D^{k}(w)$, then, by Remark \ref{rem:D}, there would exist a proper factor of $w$ which is not differentiable, against the definition of minimal forbidden word. 
\end{proof}

So, analogously to the case of $\C$-words, we can define the \emph{height} of a word $w$ in $\CC$. This is the integer $k+1$ such that $D^{k}(w)=xxx$, for $x\in \Sigma$. 

Not surprisingly, the set $\CCk$ of minimal forbidden words for the set $\Ck$ coincides with the set of words in  $\CC$ having height not greater than $k$, as shown in the following lemma.

\begin{lemma}\label{lem:mfw}
For any $k>0$, the subset of $\CC$ of words having height less than or equal to $k$ is the set $\CCk$. 
\end{lemma}

\begin{proof}
 By induction on $k$. Let first $k=1$. We have to prove that the set of minimal forbidden words for the set of $1$-differentiable words is equal to the set of minimal forbidden words of height $1$.
 By definition, a minimal forbidden word of height $1$ is a word $w$ such that $w$ is not differentiable but every proper factor of $w$ is. This directly leads to $\MF(\textbf{C}^{1})=\{111,222\}$, and  proves the basis step of the induction.
 
 Suppose now that the claim holds true for $k\ge 1$. By definition, $\CCkplus$ is the set of words $w$ such that $w$ is not ($k+1$)-differentiable, but every proper factor of $w$ is. Clearly, since a ($k+1$)-differentiable word is also a $k$-differentiable word, we have $\CCk\subset \CCkplus$. By induction hypothesis, $\CCk$ is the set of minimal forbidden words of height less than or equal to $k$. It remains to prove that every word in $\CCkplus\setminus \CCk$ has height equal to $k+1$. Let $w\in \CCkplus\setminus \CCk$. Then $w$ is $k$-differentiable but not $k+1$ differentiable. Then, by definition, $D^{k}(w)$ contains $xxx$ as factor, for some letter $x\in \Sigma$. By the minimality of $w$, it follows that $D^{k}(w)=xxx$, and hence $w$ has height $k+1$.
\end{proof}

The following lemma gives a constructive characterization of  the sets $\CCk$.

\begin{lemma}\label{lem:P-constr}
Let $k>0$ and $P_{k+1}$ be the set of words $v$ such that $v$ is a primitive of minimal length of $u$ and $u$ is a minimal forbidden word of height $k$. Then one has
$$\CCkplus=\CCk \cup P_{k+1}.$$
\end{lemma}

\begin{proof}
By Lemma \ref{lem:mfw}, the set $\CCkplus\setminus \CCk$ is the set of minimal forbidden words having height equal to $k+1$, so its elements are primitives of words in $\CCk$. By minimality, they must be primitives of minimal length. 
\end{proof}

\begin{remark}
Since every minimal forbidden word of height $k$ gives exactly two minimal forbidden words of height $k+1$ (one being the complement of the other) we have, for any $k>0$, $|\CCk|=\sum_{i=1}^{k}2^{i}=2^{k+1}-2$.
\end{remark}

The sets $\CCk$, for the first values of $k$, are reported in Table \ref{tab:mfw}.

Another characterization of minimal forbidden words is the following.

\begin{lemma}\label{lem:charmfw}
 The word $xuy$, $x,y\in \Sigma$, $u\in \Sigma^{*}$, belongs to $\CC$ if and only if the word $\overline{x}u\overline{y}$ is a minimal $\C$-word and has root $1$. 
 \end{lemma}

\begin{proof}
Suppose that $xuy$, $x,y\in \Sigma$, $u\in \Sigma^{*}$, belongs to $\CC$. Then, by Lemma \ref{lem:height}, there exists $j$ such that $D^{j}(xuy)=zzz$ for a $z\in \Sigma$. By Remark \ref{rem:mfw}, $u$ is a single-rooted maximal word. By Lemma \ref{lem:lengthD}, then, $D^{j}(\overline{x}u\overline{y})=\overline{z}z\overline{z}$, and thus $D^{j+1}(\overline{x}u\overline{y})=1$. Finally, $\overline{x}u\overline{y}$ is a minimal word since $u$ is a maximal words (Lemma \ref{lem:maxmin1}).

Conversely, let $\overline{x}u\overline{y}$ be a minimal $\C$-word of root $1$. Hence there exists $j>0$ such that $D^{j}(\overline{x}u\overline{y})=\overline{z}z\overline{z}$, for a $z\in \Sigma$. By Lemma \ref{lem:lengthD}, we have that $D^{j}(xuy)=zzz$, so $w\notin \C$. Moreover, since $\overline{x}u\overline{y}$ is a minimal word, the word $u$ is maximal (Lemma \ref{lem:maxmin1}), and then, by Theorem \ref{theor:Weakley}, $u$ is left doubly extendable and right doubly extendable. Therefore, both $uy$ and $xu$ are $\C$-words, proving thus that $xuy$ is a minimal forbidden word.
\end{proof}





\begin{table}[h]
\begin{center}
  \begin{tabular}{| c | c | c | c | c |}
  
    $\MF(\textbf{C}^{1})$   & $\MF(\textbf{C}^{2})$ &  $\MF(\textbf{C}^{3})$ \\    \hline
          &           & 111             \\
          &           & 222             \\   
          &           & 21212           \\
          &           & 12121           \\
          &   111     & 112211          \\
          &   222     & 221122          \\
     111  &   21212   & 11211211        \\
     222  &   12121   & 22122122        \\
          &   112211  & 212212212       \\
          &   221122  & 121121121       \\
          &           & 2121122121      \\
          &           & 1212211212      \\
          &           & 1122121122      \\
          &           & 2211212211      \\

    \hline
  \end{tabular}
\end{center}\caption{The sets of minimal forbidden words for $\textbf{C}^{1}$, $\textbf{C}^{2}$ and $\textbf{C}^{3}$.}\label{tab:mfw}
\end{table}

\begin{lemma}\label{lem:PrefMFW}
A $\C$-word is a proper prefix of some word in $\CC$ if and only it is a left minimal word.
\end{lemma}

\begin{proof}
Suppose that $v$ is a proper prefix of $w=xuy\in \CC$, $x,y\in \Sigma$. The word $u$ is a maximal word (Remark \ref{rem:mfw}). By Lemma \ref{lem:maxmin1}, $xu$ is then a left minimal word, and thus $v$, which is a prefix of $xu$, is a left minimal word. So the direct part of the statement is proved.

Conversely, let $v$ be a left minimal word. We prove that $v$ is a proper prefix of a minimal forbidden word by induction on $n=|v|$. The words $v=1$ and $v=2$ are proper prefixes respectively of $111$ and $222$, both belonging to $\CC$. So suppose that the claim holds true for every left minimal word of length smaller than $n>0$ and let $v$ be a left minimal word of length $n$. Consider the word $D(v)$. By Theorem \ref{theor:Weakley}, $D(v)$ is a left minimal word. Since $|D(v)|<|v|$, by inductive hypothesis $D(v)$ is a proper prefix of some minimal forbidden word $w$. The two shortest primitives of $w$ are minimal forbidden words (by Lemma \ref{lem:P-constr}). Denote them by $w'$ and $\overline{w'}$. For the direct part of the statement, the proper prefixes of $w'$ and $\overline{w'}$ are left minimal words. Now, $v$ must be a prefix of either $w'$ or $\overline{w'}$, and this completes the proof.
\end{proof}


For every $k>0$, let $\mathcal{T}(k)$ be the trie recognizing the anti-factorial language $\CCk$. We denote by $\mathcal{A}(k)$ the automaton constructed by the procedure {\sc L-automaton} on input $\mathcal{T}(k)$.

\begin{figure}
\begin{center}
\includegraphics[height=90mm]{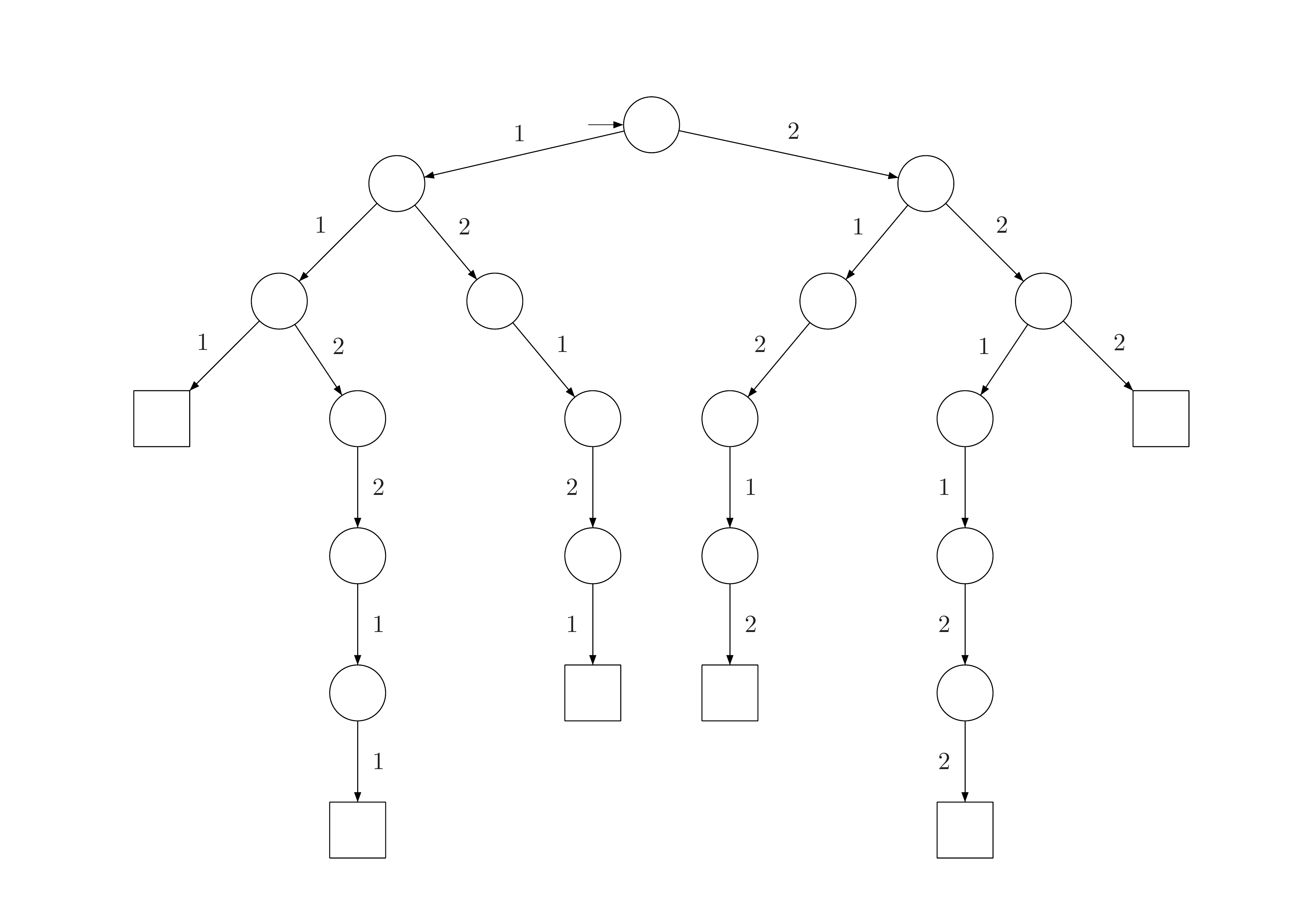}
\caption{The trie of $\MF(\textbf{C}^{2})$.}
\label{fig:Trie}
\end{center}
\end{figure}

\begin{theorem}\label{theor:Ak}
For every $k>0$, $\mathcal{A}(k)$ is a deterministic automaton recognizing the language $\Ck$.
\end{theorem}

\begin{proof}
The automaton $\mathcal{A}(k)$ is deterministic by construction. The fact that  $\mathcal{A}(k)$ recognizes the language $\Ck$ is a direct consequence of Corollary \ref{cor:Laut} and Lemma \ref{lem:P-constr}.
\end{proof}

\begin{figure}
\begin{center}
\includegraphics[height=90mm]{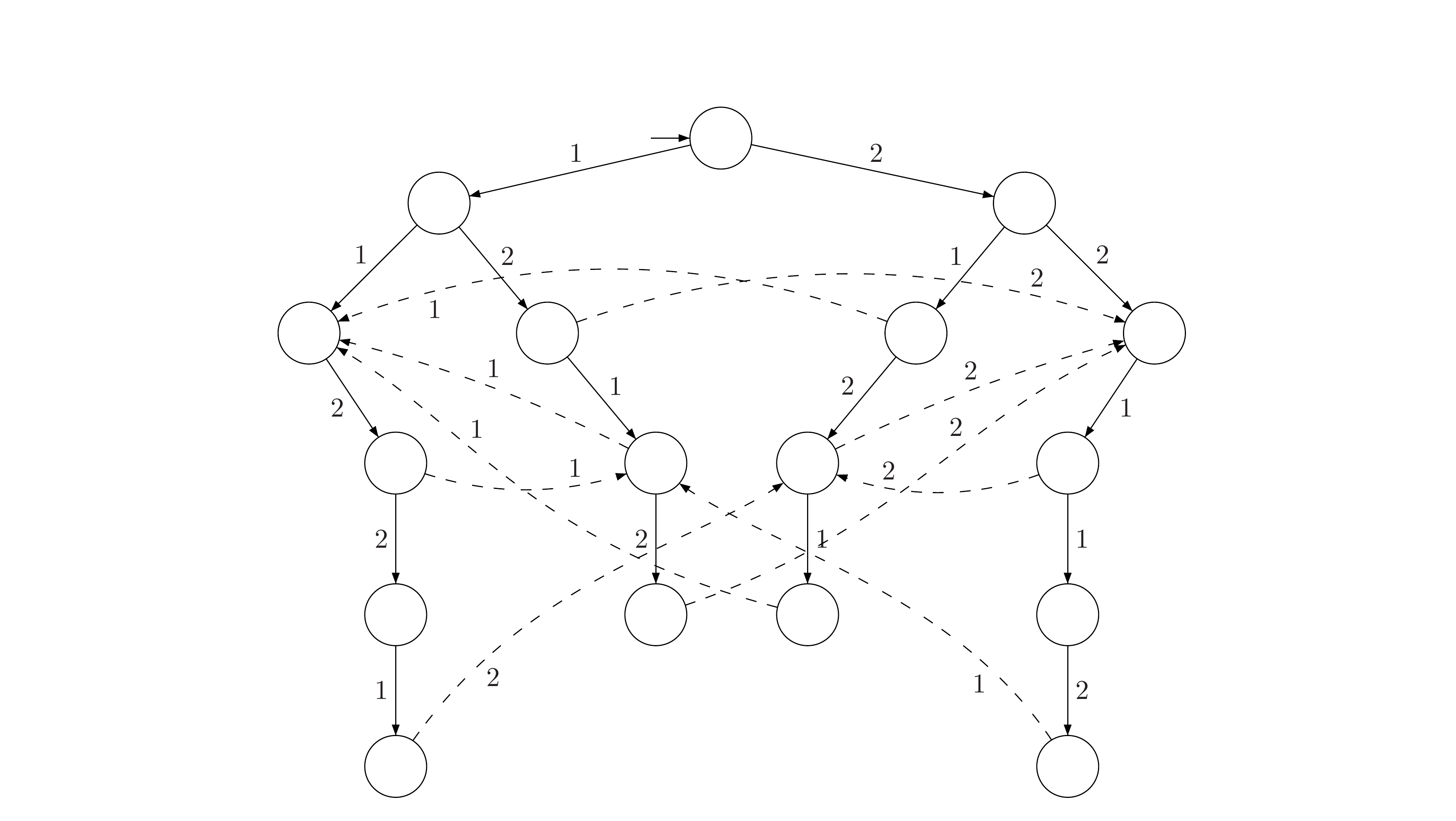}
\caption{The automaton $\mathcal{A}_{2}$ which recognizes the set $\textbf{C}^{2}$. All  states are terminal.}
\label{fig:Laut}
\end{center}
\end{figure}

Since, by Lemma \ref{lem:P-constr}, the construction of $\CCkplus$ from $\CCk$ is effective, we can inductively  extend the construction of the automaton $\mathcal{A}(k)$ to the case $k=\infty$. For this, consider the infinite trie $\mathcal{T}_{\infty}$ corresponding to the set $\CC$. Procedure {\sc L-automaton} on input $\mathcal{T}_{\infty}$ gives an infinite automaton $\A$ recognizing the words in $\C$, in the sense that any word in $\C$ is the label of a unique path in $\A$ starting at the initial state.

Let $\delta_{\mathcal{A}}$ denote the transition function of $\A$. As we already mentioned in Section \ref{sec:background}, the automaton $\A$ induces a natural equivalence on $\C$, defined by
$$u\equiv_{\mathcal{A}} v \hspace{4mm}\Longleftrightarrow \hspace{4mm} \delta_{\mathcal{A}}(\epsilon,u)=\delta_{\mathcal{A}}(\epsilon,v).$$ 
The class of a word $u$ with respect to the equivalence above will be denoted by $[u]_{\mathcal{A}}$. Let $u$ be a state of $\A$. Since $u$ is the shortest element of its class $[u]_{\mathcal{A}}$, we have that $u$ is a proper prefix of a word in $\CC$, and thus, by Lemma \ref{lem:PrefMFW}, $u$ is a left minimal word.

\begin{proposition}\label{prop:classA}
Let $u$ be a state of $\A$. Let $v\in \C$. Then $u\equiv_{\mathcal{A}} v$ if and only if $v$ is a left simple extension of $u$. 
\end{proposition}

\begin{proof}
Let $v\in \C$. By construction, the state $u=\delta(\epsilon, v)$ is the longest suffix of $v$ which is also a state of $\mathcal{T}_{\infty}$, and so, by Lemma \ref{lem:PrefMFW}, $u$ is the longest suffix of $v$ which is a left minimal word. 
Suppose by contradiction that $v$ is not a left simple extension of $u$. This implies that there exists a suffix $u'$ of $v$ such that $|u'|\geq |u|$ and $1u'$ and $2u'$ are both  $\C$-words. By Theorem \ref{theor:Weakley}, $u'$ is a left maximal word. By Lemma \ref{lem:maxmin1}, this would imply that $v'$ has a suffix $xu'$, $x\in \Sigma$, which is a left minimal word longer than $u$. The contradiction then comes from Lemma \ref{lem:classes} and Lemma \ref{lem:PrefMFW}.
\end{proof}

\begin{corollary}
Let $u$ be a left minimal word. Then 
$$[u]_{\mathcal{A}}=\{\Suff(v)\cap \Sigma^{\geq |u|},\mbox{ where $v$ is the left maximal extension of $u$}\}.$$
\end{corollary}

We now describe the transitions of $\A$. Let $u$ and $v$ be two states of $\A$. By Lemma \ref{lem:PrefMFW}, $u$ and $v$ are left minimal words. Let $x\in \Sigma$. If $(u,x,v)$ is a solid edge, then clearly $v=u x$, by definition. The weak edges, created by procedure {\sc L-automaton}, are instead characterized by the following proposition.

\begin{proposition}\label{prop:weak}
Let $u$ and $v$ be two states of $\A$ and let $x\in \Sigma$. If the transition $(u,x,v)$ is a weak edge, then:

\begin{enumerate}
 \item $u$ is a left minimal and right maximal word and is double-rooted;
 \item $v$ is a minimal word and has root $2$.
\end{enumerate}
\end{proposition}

\begin{proof}
By Lemma \ref{lem:PrefMFW}, $u$ and $v$ are left minimal words. By procedure {\sc L-automaton}, since the transition $(u,x,v)$ is a weak edge, the word $ux$ is not a word in the trie $\mathcal{T}_{\infty}$. So $ux$ is not a proper prefix of a minimal forbidden word. Then, by Lemma \ref{lem:PrefMFW}, $ux$ is not left minimal. 

Let us prove that $u$ is right maximal. By contradiction, if $u$ were not right maximal, then by Theorem \ref{theor:Weakley}, $ux$ would be a right simple extension of $u$, and so $u\overline{x}\notin \C$. This would imply that $ux$ is a word in the trie $\mathcal{T}_{\infty}$, and then that the transition $(u,x,v)$ is a solid edge, a contradiction.

We now prove that $u$ is double rooted. Suppose by contradiction that, for an integer $k>0$, one has $D^{k}(u)=y\in \Sigma$. By Lemma \ref{lem:lengthD}, then, $D^{k}(ux)=yy$ or $D^{k}(ux)=y\overline{y}$. In both cases, since $u$ is left minimal, this would imply that $ux$ is also left minimal. But $ux$ is not a state of $\A$, since the transition $(u,x,v)$ is a weak edge. So, by Lemma \ref{lem:PrefMFW}, $ux$ cannot be a left minimal word. 

The word $v$ is left minimal because it is a state of $\A$ (by Lemma \ref{lem:PrefMFW}). Moreover $v$ is a suffix of $ux$, by Lemma \ref{lem:classes}, and $ux$ is a right minimal word, since $u$ is a right maximal word (Lemma \ref{lem:maxmin1}). So $v$ is also right minimal, and then it is a minimal word. 

It remains to prove that $v$ has root equal to $2$. Since $u$ has been proved to be double-rooted, there exists $k>0$ such that $D^{k}(u)=y\overline{y}$, for a $y\in \Sigma$. By Lemma \ref{lem:lengthD}, we have that $D^{k}(ux)=y\overline{y}y$ or $D^{k}(ux)=y\overline{y}\overline{y}$. In the first case $ux$ would be a left minimal word (and we proved that this is not possible), so the second case holds. Since $v$ is the longest suffix of $ux$ which is also a left minimal word, it follows that $D^{k}(v)=\overline{y}\overline{y}$, and so $D^{k+1}(v)=2$.
\end{proof}



We can compact the automaton $\A$ by using a standard method for compacting automata, described below. We obtain a \emph{compacted version} of $\A$, denoted $\CA$.

Let $u$ be a state of $\A$ such that $u$ is a right minimal word (and thus $u$ is a minimal word since, by Lemma \ref{lem:PrefMFW}, $u$ is also a left minimal word). Let $ux_{1}x_{2}\cdots x_{n}$, $x_{i}\in \Sigma$, be the right maximal extension of $u$. This means that for every $i$, $1\leq i<n$, the transition $(ux_{1}x_{2}\cdots x_{i},x_{i+1},ux_{1}x_{2}\cdots x_{i+1})$ is the unique edge outgoing from $ux_{1}x_{2}\cdots x_{i}$ in $\A$.  The procedure for obtaining $\CA$ from $\A$ consists in identifying the states belonging to the right maximal extensions of right minimal words. For each right minimal word $u$ in $\mathcal{T}_{\infty}$, identify all the states of its right maximal extension, and replace the transitions of the right maximal extension of $u$ with a single transition $(u,x_{1}x_{2}\cdots x_{n},ux_{1}x_{2}\cdots x_{n})$ labeled with the concatenation of the labels of the transitions in the right maximal extension of $u$. In this way, there are exactly two edges outgoing from each state (either two solid edges or one solid edge and one weak edge).

If in $\A$ there is a weak edge $(u,x,v)$, and $v'$ is the right maximal extension of $v$, then in $\CA$ there will be a weak edge $(u,x,v')$. The label of this weak edge is then set to be the same word labeling the (unique) solid edge ingoing to $v'$ in $\CA$.

A partial diagram of the automaton $\CA$ is depicted in Figure \ref{fig:CA}.

\begin{figure}
\begin{center}
\includegraphics[height=180mm]{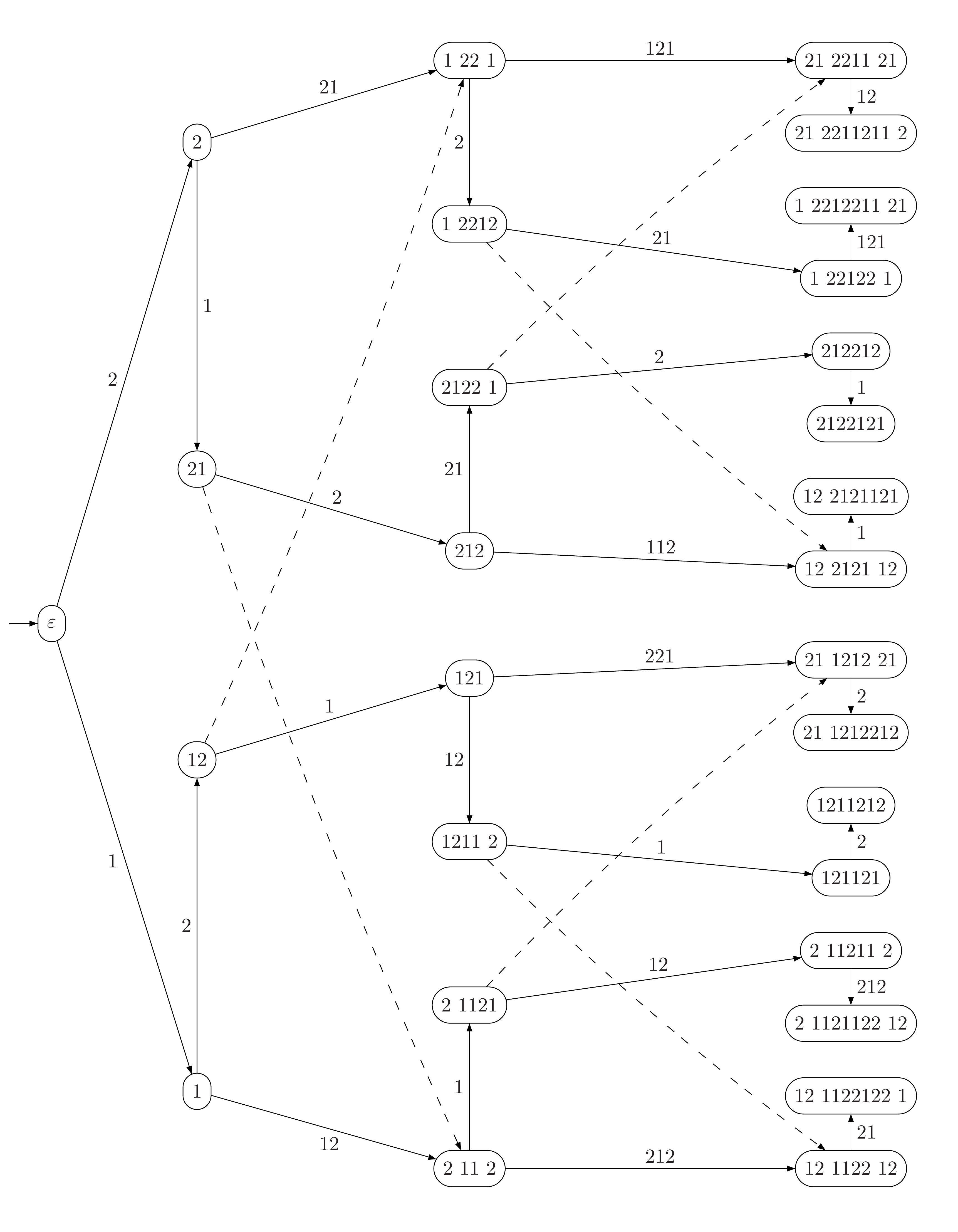}
\caption{The automaton $\CA$ cut at height 3. The labels of weak edges are omitted. All  states are terminal.}
\label{fig:CA}
\end{center}
\end{figure}

The automaton $\CA$ induces on $\C$ a new equivalence, defined by
\begin{center}
\begin{tabular}{p{30mm}l@{ }}
$u\equiv_{\mathcal{CA}} v\hspace{4mm} \Longleftrightarrow$   & $\delta_{\mathcal{A}}(\epsilon,u)$ and $\delta_{\mathcal{A}}(\epsilon,v)$ belong to the right \\
       &     maximal extension of the same state.
\end{tabular} 
\end{center}

\begin{proposition}\label{prop:CAclass}
Let $u,v\in \C$. Then $u\equiv_{\mathcal{CA}} v$ if and only if $u$ and $v$ have the same maximal extension.
\end{proposition}

\begin{proof}
Let $u,v\in \C$ and let $u'=\delta(\epsilon,u), v'=\delta(\epsilon,v)$. Then, by Proposition \ref{prop:classA}, $u$ is a left simple extension of $u'$ and $v$ is a left simple extension of $v'$. On the other hand, by definition of $\CA$, $u'$ and $v'$ are right simple extensions of the same word $w$, that, by Lemma \ref{lem:PrefMFW}, is a left minimal word. Thus, $u$ and $v$ are left simple extensions of right simple extensions of $w$, i.e., they are both factors of the maximal extension of $w$.
\end{proof}

So,  each state of $\CA$ can be identified with a class of words having the same maximal extension. We denote by $[u]_{\mathcal{CA}}$ the class of $u$ with respect to the equivalence $\equiv_{\mathcal{CA}}$.  Every class $[u]_{\mathcal{CA}}$ contains a unique shortest element $u$, which is a minimal word. The other elements in $[u]_{\mathcal{CA}}$ are the simple extensions of $u$, to the left and to the right, up to the maximal extension of $u$, which is a maximal word (by Remark \ref{rem:extension}). By Lemma \ref{lem:ext}, all the words belonging to the same class with respect to the equivalence $\equiv_{\mathcal{CA}}$ have the same height and the same root. Therefore, we can unambiguously define the height and the root of a state in $\CA$.

\begin{remark}\label{rem:CAstates}
 For every $k>0$, there are $2^{k}$ states of height $k$ in $\CA$. In particular, there are $2^{k-1}$ single-rooted states and $2^{k-1}$ double-rooted states.
\end{remark}

\begin{proposition}\label{prop:CAedge}
 Let $u$ be a state of $\CA$. If $u$ is single rooted, then there are two solid edges outgoing from $u$. If instead $u$ is double-rooted, then there are one solid edge and one weak edge outgoing from $u$.
\end{proposition}

\begin{proof}
The claim follows from the construction of $\CA$ and from Proposition \ref{prop:weak}.
\end{proof}

\section{Vertical representation of $\C$-words}\label{sec:vertical}

In this section, we introduce a new framework for dealing with $\C$-words. We define a function $\Psi$ for representing a $\C$-word on a three-letter alphabet $\Sigma_{0}=\{0,1,2\}$.  
This function is a generalization of the function $\Phi$ considered in \cite{BrLa03}, that associates to any $\C$-word $w=w[0]w[1]\cdots w[n-1]$ the sequence of the first symbols of the derivatives of $w$, that is, the function defined by $\Phi(w)[i]=D^{i}(w)[0]$ for $0\le i<k$, where $k$ is the height of $w$. 

If one takes the first and the last symbol of each derivatives of a $\C$-word $w$, that is, the pair $\Phi(w),\Phi(\tilde{w})$, one gets a representation of $\C$-words that is not injective. For example, take the two $\C$-words $w=2211$ and $w'=21121221$. Then one has $\Phi(w)=\Phi(w')=222$ and $\Phi(\tilde{w})=\Phi(\widetilde{w'})=122$. In order to obtain an injective representation, we need an extra symbol. We thus introduce the following definition.

\begin{definition}
Let  $w=w[0]w[1]\cdots w[n-1]$ be a $\C$-word of height $k>0$. The \emph{left frontier} of $w$ is the word $\Psi(w)\in \Sigma_{0}^{k}$ defined by $\Psi(w)[0]=w[0]$ and for $0<i<k$
$$\Psi(w)[i] = \left\{ \begin{array}{lllll}
0 & \mbox{if $D^{i}(w)[0]=2$ and $D^{i-1}(w)[0]\neq D^{i-1}(w)[1]$,}\\
D^{i}(w)[0] & \mbox{otherwise.}
\end{array} \right.$$
For the empty word, we set $\Psi(\varepsilon)=\varepsilon$.

The \emph{right frontier} of $w$ is defined as $\Psi(\tilde{w})$. If $U$ and $V$ are respectively the left and right frontier of $w$, we call $U|V$ the \emph{vertical representation} of $w$. 
\end{definition}

In other words, to obtain the left (resp.\ the right) frontier of $w$, one has to take the first (resp.\ the last) symbol of each derivative of $w$ and replace a $2$ by a $0$ whenever the primitive above is not left minimal (resp.\ is not right minimal). 

\begin{example}\label{ex:vert}
Let $w=21221211221$. We have:
\[ \begin{tabular}{p{15mm} l}
\hline \rule[-6pt]{0pt}{18pt}$D^{0}(w)$  & $21221211221$ \\
\rule[-6pt]{0pt}{13pt}$D^{1}(w)$  & $121122$ \\
\rule[-6pt]{0pt}{13pt}$D^{2}(w)$  & $122$\\
\rule[-6pt]{0pt}{13pt}$D^{3}(w)$  & $2$\\
\hline \rule[-2pt]{0pt}{2pt}
\end{tabular}  \]
The word $D^{2}(w)=122$ is not a left minimal primitive of the word $D^{3}(w)=2$, and therefore $\Psi(w)[3]$, the fourth symbol of the left frontier of $w$, is a $0$; analogously, the word $w=21221211221$ is not a right minimal primitive of $D(w)=121122$, and therefore $\Psi(\tilde{w})[1]$, the second symbol of the right frontier of $w$, is a $0$. Hence, the vertical representation of $w$ is $\Psi(w)|\Psi(\tilde{w})=2110|1022$.
\end{example}

\begin{remark}
By definition, for any $\C$-word $w$ of height $k>0$, we have that $\Psi(w)=\Psi(w)[0]\Psi(w)[1]\cdots \Psi(w)[k-1]$ is a word of length $k$ over $\Sigma_{0}$ whose first symbol is different from $0$. Conversely, any word $U$ of length $k>0$ over $\Sigma_{0}$, such that its first symbol is different from 0, is the left frontier of some $\C$-word of height $k$.
\end{remark}

\begin{theorem}\label{theor:vertical}
Any word in $\C$ is uniquely determined by its vertical representation. 
\end{theorem}

\begin{proof}
The claim  follows directly from the definition of $\Psi$.
\end{proof}

\begin{remark}\label{rem:vert}
Let $w$ be a $\C$-word of height $k>0$ and $U|V$ its vertical representation. Then:

\begin{enumerate}
\item $w$ is left (resp.\ right) maximal if and only if $U[i]\neq 2$ (resp.\ $V[i]\neq 2$) for every $i=1,\ldots, k-1$;

\item $w$ is left (resp.\ right) minimal if and only if $U[i]\neq 0$ (resp.\ $V[i]\neq 0$) for every $i=1,\ldots, k-1$. 
 
\end{enumerate}
\end{remark}

We shall explore further properties of the vertical representation in a forthcoming paper \cite{FeFi12}.


The \emph{vertical compacted automaton}, noted $\mathcal{VCA}_{\infty}$, is obtained from $\mathcal{CA}_{\infty}$ by replacing the label of each state $u$ by the vertical representation of $u$, and by replacing the labels of the transitions in the following way: a solid edge from a state $U|V$ to a state $Ux|V'$, $x\in \Sigma$, is labeled by $x$; a solid edge from a state $U|V$ to a state $U|V'$ is labeled by $\epsilon$; finally, weak edges are labeled by $0$. 

A partial diagram of automaton $\VCA$ is depicted in Figure \ref{fig:VCA}.

\begin{figure}
\begin{center}
\includegraphics[height=180mm]{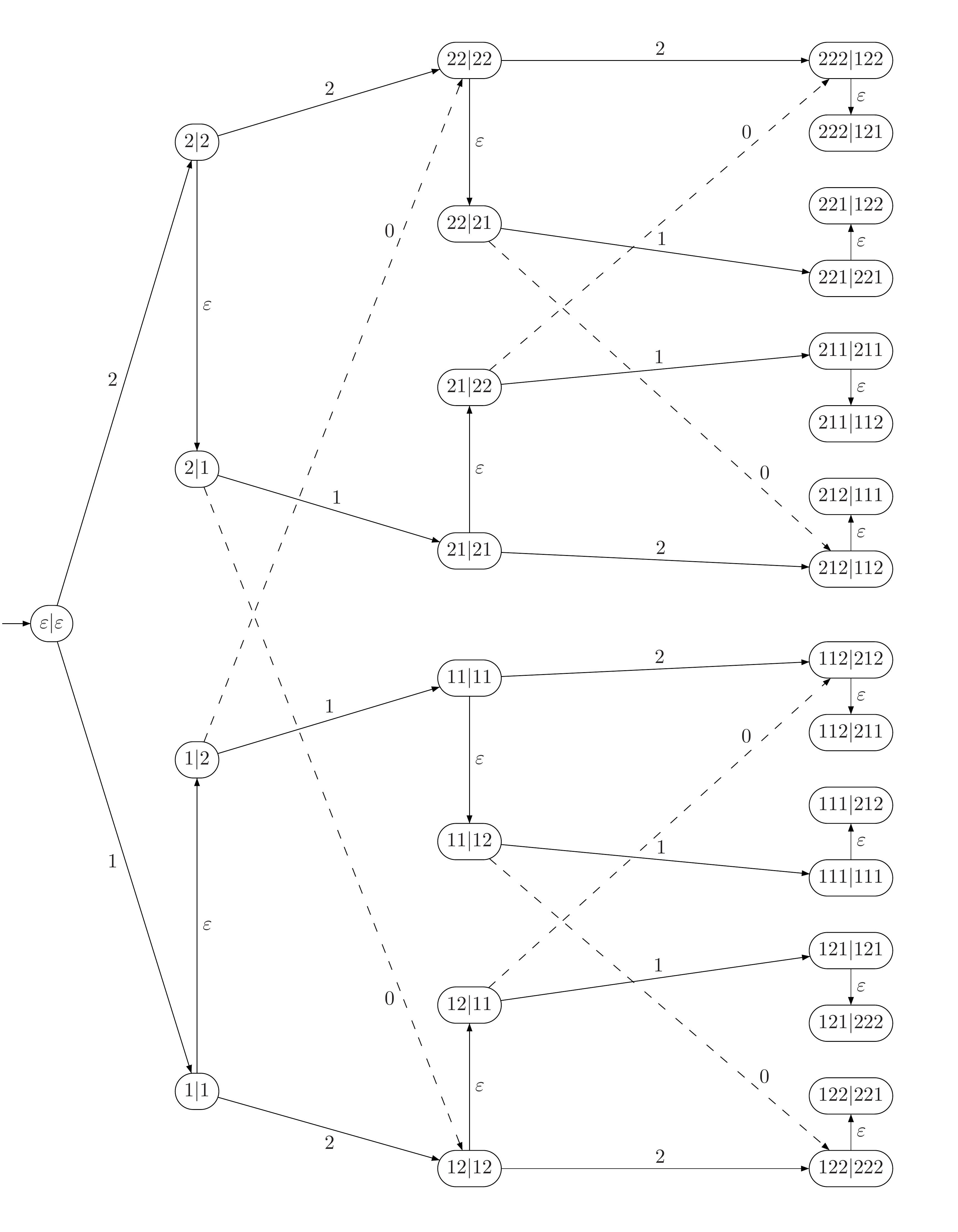}
\caption{The automaton $\VCA$ cut at height 3. All states are terminal.}
\label{fig:VCA}
\end{center}
\end{figure}

The choice of introducing $\epsilon$-transitions is motivated by the following considerations. By construction, each state of the automaton $\CA$ corresponds to a class of words having the same height and the same root. There is a unique minimal word $w$ in each state, and all other words in the same state are the simple extensions, on the left and on the right, of $w$. There are two kinds of minimal words: single-rooted and double-rooted. It is easy to see that for every $U=U[0]U[1]\cdots U[k-1]\in \Sigma^{k}$, $U$ is the left frontier of exactly two minimal words of height $k$: a single-rooted word $w_{1}$ having root equal to $U[k-1]$ and a double-rooted word $w_{2}$ having root equal to $U[k-1]\overline{U[k-1]}$. Moreover, $w_{1}$ is a prefix of $w_{2}$. Therefore, if we label with $\epsilon$ each transition from the class of $w_{1}$ to the class of $w_{2}$, a path in $\VCA$ starting at the initial state and labeled by $U$ ends in a state having label $U|V$.

As a consequence, we can thus further compact the automaton by identifying the pairs of states in $\VCA$ that have the same left frontier. This corresponds to identify the class of $w_{1}$ with the class of $w_{2}$. In this way, each class of words is uniquely determined by the left frontier of its minimal element only. The resulting automaton, called the \emph{ultra-compacted} version of $\VCA$, is noted $\VUCA$.  The transitions are labeled by the letters of $\Sigma_{0}$. The letter $0$ is the label of the weak edges, while $1$ and $2$ label solid edges. The trie formed by the solid edges of $\VUCA$ is a complete binary tree in which there are $2^{k}$ nodes at level $k$ representing the left frontiers of the minimal words of height $k$.  Each state of $\VUCA$ different from $\epsilon$ has exactly three outgoing edges: two solid edges, labeled by 1 and  2, and a weak edge labeled by 0. 

Actually, cutting the infinite automaton $\VUCA$ at level $k$, one obtains a deterministic automaton $\mathcal{VUCA}_{k} =(Q,\Sigma_{0},\epsilon,Q,\delta_{\mathcal{VUCA}})$, where $Q=\Sigma^{\leq k}$ and $\delta_{\mathcal{VUCA}}$ is defined by:

\begin{enumerate}
 \item $\delta_{\mathcal{VUCA}}(U,x)=Ux$ if $x\in \Sigma$;
 \item $\delta_{\mathcal{VUCA}}(U,0)=V2$ if for any $u\in \C$ such that $\Psi(u)=U0$, the longest suffix $v$ of $u$ that is also a left minimal word has left frontier equal to $\Psi(v)=V2$.
\end{enumerate}

A partial diagram of automaton $\VUCA$ is depicted in Figure \ref{fig:VUCA}. Note that the order of the states at each level (the lexicographic order in the upper half, and the reverse of the lexicographic order in the lower half) makes the graph of the automaton symmetric. This property follows from the symmetry of the vertical representation of $\C$-words with respect to the swap of the first symbol, which, in turns, represents the symmetry of $\C$-words with respect to the complement. Indeed, if a word $w$ has left frontier $U[0]U[1]\cdots U[k-1]$, then the word $\overline{w}$ has left frontier $\overline{U[0]}U[1]\cdots U[k-1]$.

\begin{figure}
\begin{center}
\includegraphics[height=180mm]{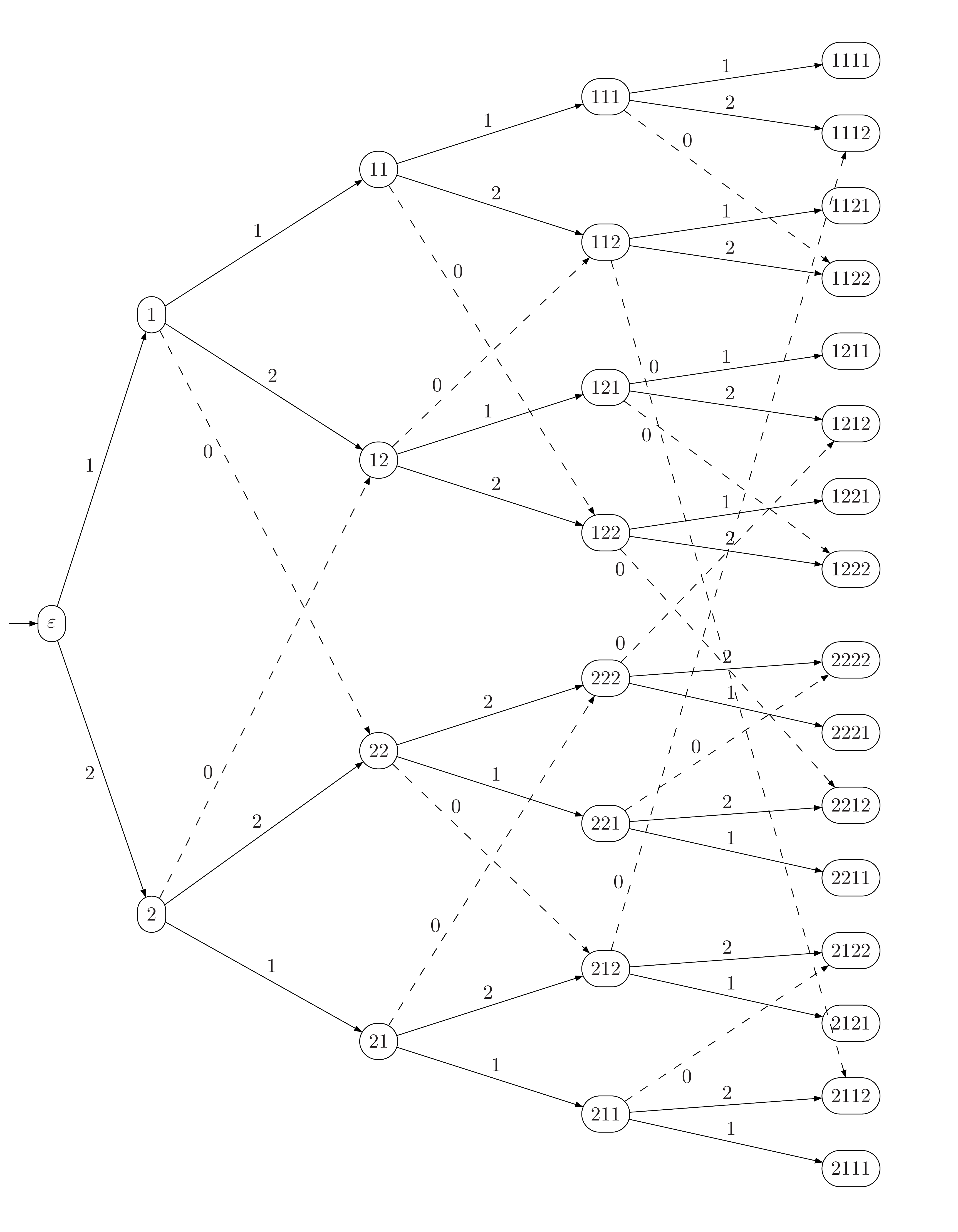}
\caption{The automaton $\VUCA$ cut at height 4. All states are terminal. The order of the states at each level is the lexicographic order in the upper half, and the reverse of the lexicographic order in the lower half. This makes the graph of the automaton symmetric.}
\label{fig:VUCA}
\end{center}
\end{figure}

The automaton $\VUCA$ induces on the set of $\C$-words a natural equivalence defined by
$$u\equiv_{\mathcal{VUCA}} v\hspace{4mm}\Longleftrightarrow \hspace{4mm}\delta_{\mathcal{VUCA}}(\epsilon,\Psi(u))=\delta_{\mathcal{VUCA}}(\epsilon,\Psi(v)).$$

We denote the class of $u$ with respect to this equivalence by $[u]_{\mathcal{VUCA}}$.

\begin{proposition}\label{prop:classVUCA}
Let $u,v\in \C$. Then $u\equiv_{\mathcal{VUCA}} v$ if and only if the left maximal extension of $u$ and the left maximal extension of $v$ have the same left frontier.
\end{proposition}

\begin{proof}
 The claim is a direct consequence of the construction of $\VUCA$.
\end{proof}

We end the section by discussing an interesting property of the automaton $\VUCA$.
Let $w$ be a $\C$-word. By Theorem \ref{theor:vertical}, $w$ is uniquely determined by its vertical representation $\Psi(w)|\Psi(\tilde{w})$. Moreover, $w$ is a simple extension of a unique minimal word $w'$ having the same height and the same root as $w$ (Lemma \ref{lem:ext}). To get the vertical representation $\Psi(w')|\Psi(\widetilde{w'})$ of the word $w'$, one can use the automaton $\VUCA$. Indeed, $\Psi(w)$ is the label of a unique path in $\VUCA$ starting at the origin and ending in a state $U$. Then $U$ is the left frontier of $w'$, i.e., $U=\Psi(w')$. Analogously, $\Psi(\tilde{w})$ is the label of a unique path in $\VUCA$ starting at the origin and ending in a state $V$, and $V$ is the right frontier of $w'$, i.e., $V=\Psi(\widetilde{w'})$.

\begin{example}
 Let $w= 21221211221$ as in Example \ref{ex:vert}. The vertical representation of $w$ is $2110|1022$. Looking at the graph of the automaton $\VUCA$ (Figure \ref{fig:VUCA}) we see that the path starting at the origin and labeled by $2110$ ends in state $2122$, while the path starting at the origin and labeled by $1022$ ends in state $2222$. Thus, the minimal word of which $w$ is a simple extension is the word $w'$ having vertical representation $2122|2222$, that is, the word $w'=2121122$.
\end{example}

\begin{example}
 Let $w= 1221221121$. The vertical representation of $w$ is $101|110$. Looking at the graph of the automaton $\VUCA$ (Figure \ref{fig:VUCA}) we see that the path starting at the origin and labeled by $101$ ends in state $221$, while the path starting at the origin and labeled by $110$ ends in state $122$. Thus, the minimal word of which $w$ is a simple extension is the word $w'$ having vertical representation $221|122$, that is, the word $w'=2212211$.
\end{example}

\section{$\C$-words of the form $uzu$}\label{sec:uvu}

In this section, we use the structure of $\VUCA$ for deriving an upper bound on the length of the gap between two occurrences of a $\C$-word. Recall that a \emph{repetition with gap} $n$ of the $\C$-word $u$ is a $\C$-word of the form $uzu$ such that $|z|=n$. Carpi \cite{Carpi:1994} proved that for every $n>0$ there are finitely many repetitions with gap $n$ in $\C$. In a more recent paper, Carpi and D'Alonzo \cite{CaDa09} proved that the repetitivity index of $\C$-words is ultimately bounded from below by a linear function. The \emph{repetitivity index} \cite{CaDa09bis} is the integer function $I$ defined by
$$I(n)=\min\{k>0 \mbox{ $|$ } \exists u\in \C, |u|=n \mbox{ : $uzu\in \C$ for a $z$ such that $|z|=k$}\}.$$ In other words, $I(n)$ gives the minimal gap of a repetition of a word $u$ of length $n$ in $\C$.

We now explore the relationship between the length and the height of a $\C$-word. For any $\C$-word $w$, we have 
$$|D(w)|+2|D(w)|_{2}\leq |w|\leq |D(w)|+2|D(w)|_{2}+2.$$
Chv\'atal \cite{Chvatal:1993} proved that the upper density of $2$'s in a $k$-differentiable word, for $k>22$, is less than $p=0.50084$.
Hence, we can suppose that for every $\C$-word $w$ of height $k>22$ one has 
\begin{equation}\label{eq:p}
(2-p)|D(w)|\leq |w|\leq (1+p)|D(w)|+2.
\end{equation}
We thus have the following lemma.

\begin{lemma}
There exist positive constants $\alpha$ and $\beta$ such that for any $\C$-word $w$
\begin{equation}\label{eq:l}
\alpha(2-p)^{h(w)} <|w|< \beta(1+p)^{h(w)},
\end{equation}
and therefore
\begin{equation}\label{eq:h}
\frac{\log |w|-\log \beta}{\log(1+p)} < h(w) < \frac{\log |w| - \log \alpha}{\log(2-p)},
\end{equation}
where $h(w)$ is the height of $w$.
\end{lemma}

\begin{theorem}\label{theor:uzu}
Let $u\in\C$. Then there exists $z\in \C$ such that $uzu\in \C$ and $|uzu|\leq C|u|^{2.72}$, for a suitable constant $C$.
\end{theorem}

\begin{proof}
Let $u$ be a $\C$-word of height $h(u)$. Without loss of generality, we can suppose that $u$ is a maximal word. Indeed, if a word $w'$ is the maximal extension of a word $w$, then a word of the form $w'z'w'$, $z'\in \Sigma^{*}$, contains a word of the form $wzw$ as factor, for a $z\in \Sigma^{*}$. Moreover, $|wzw|\leq |w'z'w'|$.

Let $U'=\Psi(u)$ be the left frontier of the word $u$. Then $U'$ is the label of a unique path in $\VUCA$ starting at the initial state and ending in a state $U$. Consider the paths in $\VUCA$ outgoing from the state $U$. Since each state of $\VUCA$ has exactly three outgoing edges, there are $3^{n}$ distinct paths of length $n$ starting in $U$. Each of these paths ends in a state $W$ such that $|W|=h(u)+n$. Since there are $2^{n+h(u)}$ distinct states $W$ such that $|W|=h(u)+n$, by the pigeonhole principle there will be two distinct paths of length $n$ starting in $U$ and ending in the same state, say $V$, whenever $3^{n}>2^{n+h(u)}$, that is, whenever $n>\gamma h(u)$, where $\gamma=(\log_{2} 3-1)^{-1}\simeq 1.70951$.

So there exists a state $V$ in $\VUCA$ such that $|V|\leq \lceil (1+\gamma)h(u)\rceil$ and there are two distinct paths, say $V_{1}$ and $V_{2}$, from $U$ to $V$. 

Thus, there exist two distinct $\C$-words $v_{1}$ and $v_{2}$ such that $\Psi(v_{1})=U'V_{1}$ and $\Psi(v_{2})=U'V_{2}$ (and this implies that $u$ is prefix of both $v_{1}$ and $v_{2}$), and $v_{1}\equiv_{\mathcal{VUCA}} v_{2}$. Moreover, if $v$ is a $\C$-word such that $\Psi(v)=V$, we can suppose that $v_{1}$ and $v_{2}$ are two distinct left simple extensions of $v$. Hence, we can suppose that one of the two words (say $v_{1}$) is a suffix of the other ($v_{2}$). This implies that $u$ appears as a prefix of $v_{2}$ and has at least a second occurrence as a (proper) factor in $v_{2}$.

We suppose that the two occurrences of $u$ in $v_{2}$ do not overlap. Actually, the set $\C$ does not contain overlaps of length greater than $55$ (since every overlap contains two squares \cite{Carpi:1994}), so our assumption consists in discarding a finite number of cases, that can be included in the constant $C$ of the claim.

Thus, we can write $v_{2}=uzuv'_{2}$, for a $v'_{2}\in \Sigma^{*}$, and we have
$$h(uzu)\leq h(v_{2})\leq \lceil (1+\gamma)h(u)\rceil\leq (1+\gamma)h(u)+1.$$ 
By Equations \ref{eq:l} and \ref{eq:h}, we have
\begin{eqnarray*}
 |uzu| &\leq& \beta (1+p)^{(1+\gamma) h(u)+1}\\
       &<& \beta(1+p)(1+p)^{(1+\gamma) \frac{\log |u| -\log \alpha}{\log (2-p)}}\\
       &=& \frac{\beta (1+p)}{\alpha ^{(1+\gamma)\frac{\log (1+p)}{\log (2-p)}}}|u|^{(1+\gamma)\frac{\log (1+p)}{\log (2-p)}}\\
       &=& C|u|^{\gamma'},
\end{eqnarray*}
where $C$ is a constant and $\gamma'\simeq 2.71701$.

\end{proof}


In the proof of Theorem \ref{theor:uzu}, we do not exhibit the word $uzu$ of the claim. Actually, to obtain such a word, one has to explore (a finite portion of) the graph of $\VUCA$. We do not know whether a direct construction of the word $uzu$ is possible using another approach.

As a direct consequence of Theorem \ref{theor:uzu}, we have a sub-cubic upper bound on the length of a repetition (with gap) of a $\C$-word. 

Let us define the function
$$G(n)=\min \{k\mbox{ $|$ }\forall u\in \C, |u|=n, \exists z:|z|\leq k,uzu\in\C\}.$$ 

The function $G$ is a dual function with respect to the repetitivity index $I(n)$. As a consequence of Theorem \ref{theor:uzu}, we have:

\begin{corollary}
 $G(n)=o(n^{3})$.
\end{corollary}

\section{Conclusion}\label{sec:conclusion}

In this paper we exhibited different classifications of $\C$-words based on simple extensions, by means of graphs of infinite automata representing the set of $\C$-words. Our approach makes use of an algorithmic procedure for constructing deterministic automata, but the main interest in using this approach is that this allows us to define a structure (the graph of the infinite automaton) for representing the whole set of $\C$-words.

The vertical representation of $\C$-words introduced in Section \ref{sec:vertical} leads to a more compact automaton representing $\C$-words, $\VUCA$, keeping at the same time all the information on the words. Indeed, this novel representation allows one to manipulate $\C$-words without requiring detailed knowledge of the particular sequence of $1$s and $2$s appearing (or not) in them. In a forthcoming paper we will discuss more in depth the properties of the vertical representation of $\C$-words \cite{FeFi12}.

In Theorem \ref{theor:uzu} we gave an upper bound on the length of a repetition with gap of a $\C$-word. It is a dual result with respect to the lower bounds obtained by Carpi \cite{Carpi:1994,CaDa09}. Numerical experiments suggest that a tighter bound on the gap of a repetition of a $\C$-word $u$ could be sub-quadratic in the length of $u$. 

The proof of Theorem \ref{theor:uzu} does not allow one to build a repetition of a $\C$-word directly, that is, without using the graph of the automaton $\VUCA$. However, this is a consequence of the particular approach we used. In fact, most of the known results about the existence of particular patterns in $\C$-words make use of standard methods in combinatorics on words, while we think that the techniques we developed in this paper represent a novel approach to the study of $\C$-words. We hope that this will stimulate further developments, eventually leading to the solution of Problem \ref{probcol} and, perhaps, to a proof of at least some of the longstanding conjectures on the Kolakoski word.

\section{Acknowledgements}\label{sec:ack}

The authors are grateful to anonymous referees for their comments that greatly improved the presentation of the paper.

The second author acknowledges support from the AutoMathA program of the European Science Foundation.

\bibliographystyle{plain}


\end{document}